\documentclass[12pt]{article}

\usepackage{graphicx}
\usepackage{amsmath,amssymb,setspace,amsfonts,enumitem,fullpage,amsthm,comment,color,bm,nth}
\usepackage {palatino}
\usepackage[usenames,dvipsnames]{xcolor}
\usepackage[colorlinks,pagebackref=false]{hyperref}
\usepackage[longnamesfirst]{natbib}

\usepackage[compact,small,raggedright]{titlesec}
\titlelabel{\thetitle.\quad}
\usepackage[font={small,it},labelfont=bf,labelsep=period]{caption}
\usepackage{enumitem}
\setlist{noitemsep,parsep=0pt,partopsep=0pt,topsep=0pt}
\usepackage{todonotes}
\usepackage[nottoc]{tocbibind} 

\let \savenumberline \numberline
\def \numberline#1{\savenumberline{#1.}}

\let\oldfootnote\footnote
\renewcommand\footnote[1]{\oldfootnote{\hspace{.5mm}#1}}

\titlespacing*{\subsection}{0pt}{6pt}{0pt}

\makeatletter
\renewenvironment{proof}[1][\proofname] {\par\pushQED{\qed}\normalfont\topsep6\p@\@plus6\p@\relax\trivlist\item[\hskip\labelsep\bfseries#1\@addpunct{.}]\ignorespaces}{\popQED\endtrivlist\@endpefalse}
\makeatother

\definecolor{dark-red}{rgb}{0.4,0.15,0.15}
\definecolor{dark-blue}{rgb}{0.15,0.15,0.75}
\definecolor{medium-blue}{rgb}{0,0,0.5}

\hypersetup{
    pdftitle={Competition and Disclosure},
    pdfauthor={Kartik, Lee, Suen},
    citecolor={dark-blue},
    bookmarksnumbered=true,
    urlcolor=BrickRed,
    linkcolor=dark-blue
}

\newcounter{example}
\newenvironment{example}[1][]{\refstepcounter{example}\par\medskip\noindent%
\textbf{Example~\theexample. #1} \rmfamily}{\medskip}

\newtheorem{theorem}{Theorem}
\newtheorem{claim}{Claim}

\newtheorem{corollary}{Corollary}
\newtheorem{lemma}{Lemma}

\newtheorem{proposition}{Proposition}

\theoremstyle{remark}

\theoremstyle{definition}

\newenvironment{suppappendixref}[1]{\hyperref[#1]{Supplementary Appendix \ref{#1}}}

\renewcommand{\bar}{\overline}
\providecommand{\abs}[1]{\lvert#1\rvert}

\def \cal{\mathcal}
\def \Reals{\mathbb R}

\renewcommand{\epsilon}{\varepsilon}

\newcommand{\finishpoint}{\hfill $\diamond$}

\newcommand{\E}{\mathbb{E}}

\newcommand{\mailto}[1]{\href{mailto:#1}{\texttt{#1}}} 
\newcommand{\citepos}[1]{\citeauthor{#1}'s \citeyearpar{#1}} 

\begin{document}

\begin{titlepage}

\title{\textbf{\Large Information Validates the Prior:\\ A Theorem on Bayesian Updating and Applications}\thanks{A previous version of this paper was titled ``\href{http://www.columbia.edu/\~nk2339/Papers/KLS-disclosure.pdf}{A Theorem on Bayesian Updating and Applications} \href{http://www.columbia.edu/\~nk2339/Papers/KLS-disclosure.pdf}{to Communication Games}''. We thank Nageeb Ali, Ilan Guttman, Rick Harbaugh, Keri Hu, Alessandro Lizzeri, Steve Matthews, Paula Onuchic, Ariel Pakes, Andrea Prat, Mike Riordan, Satoru Takahashi, Jianrong Tian, Enrico Zanardo, Jidong Zhou, and various seminar and conference audiences for comments. We have benefitted from thoughtful suggestions made by the Editor and three referees. Sarba Raj Bartaula, Bruno Furtado, Dilip Ravindran, and Teck Yong Tan provided excellent research assistance. Kartik gratefully acknowledges financial support from the NSF (Grant SES-1459877).}}
\author{{Navin Kartik}\thanks{Department of Economics, Columbia University. E-mail: \mailto{nkartik@columbia.edu}.}  \and {Frances Xu Lee}\thanks{Quinlan School of Business, Loyola University Chicago. E-mail: \mailto{francesxu312@gmail.com}.} \and {Wing Suen}\thanks{Faculty of Business and Economics,
The University of Hong Kong. E-mail: \mailto{wsuen@econ.hku.hk}.}}
\date{July 21, 2020}

\maketitle

\begin{abstract}
\noindent 
We develop a result on expected posteriors for Bayesians with heterogenous priors, dubbed information validates the prior (IVP). Under familiar ordering requirements, Anne expects a (Blackwell) more informative experiment to bring Bob's posterior mean closer to Anne's prior mean. We apply the result in two contexts of games of asymmetric information: voluntary testing or certification, and costly signaling or falsification. IVP can be used to determine how an agent's behavior responds to additional exogenous or endogenous information. We discuss economic implications.
\end{abstract}

\bigskip

\thispagestyle{empty}
\end{titlepage}

\newpage

\begingroup
\singlespacing
\addtocontents{toc}{\protect\setcounter{tocdepth}{1}}
\tableofcontents
\addtocontents{toc}{~\hfill\textbf{Page}\par}
\thispagestyle{empty}
\endgroup

\newpage
\onehalfspacing
\setcounter{page}{1}

\section{Introduction}
\label{sec:intro}

Bayesian agents revise their beliefs upon receiving new 
information.  From an ex-ante point of view, however, one cannot expect information to systematically alter one's beliefs in any particular direction. More precisely, a fundamental property of Bayesian updating is that beliefs are a martingale: an agent's expectation of his posterior belief is equal to his prior belief.   But what about an agent's expectation of \emph{another} agent's posterior belief when their current beliefs differ? Relatedly, should agents expect new information to systematically affect their disagreement, and if so, how?
These questions are not only of intrinsic interest, but tackling them is useful for economics with asymmetric information.


\paragraph{Information validates the prior.} In \autoref{sec:IVP}, we develop the following general result.
Let $\Omega\subset \Reals$ be possible states. Bayesians Anne ($A$) and Bob ($B$) have mutually-known but different priors over $\Omega$, with means $m_A$ and $m_B$. A signal $s$ will be drawn from a known information structure or experiment $\mathcal E$. Let $m^s_B$ denote Bob's posterior mean after observing the signal. Let $\E_A^\cal E[m^s_B]$ be Anne's ex-ante expectation of Bob's posterior mean under experiment $\cal E$.

Consider two experiments $\mathcal E$ and $\mathcal{\tilde E}$, with $\mathcal{E}$ more informative than $\mathcal{\tilde E}$ in the sense of \citet{Blackwell51,Blackwell53}. \autoref{thm:IVP} establishes that under familiar ordering requirements,\footnote{The priors must be likelihood-ratio ordered and experiments must satisfy the monotone likelihood-ratio property. These assumptions are unrestrictive if the state is binary.}
\begin{equation}
\label{e:intro}
{m}_A \leq (\geq)\, m_B \implies 
\E_A^{\mathcal E} \left[m^s_B\right]\leq (\geq) \, \E_A^{\tilde{\mathcal E}} \left[m^{\tilde s}_B\right].
\end{equation}
In words: if Anne has a lower (resp., higher) prior mean than Bob, then Anne predicts that a more informative experiment will, on average, reduce (resp., raise) Bob's posterior mean by a larger amount than a less informative experiment. Put differently, Anne expects more information to further validate her prior in the sense of bringing Bob's posterior mean closer to her prior mean. Of course, Bob expects just the reverse. We refer to the result as \emph{information validates the prior}, IVP hereafter.

IVP has an implication about expected disagreement. A particular signal can lead to larger posterior disagreement than prior disagreement. Nevertheless, IVP implies (\autoref{cor:disagree}) that when beliefs' disagreement is quantified by the difference in means, both Anne and Bob predict that a more informative experiment will, on average, reduce their posterior disagreement by more. At the extreme, both predict zero posterior disagreement under a fully informative experiment---even though each one's expectation of the other's posterior mean may be very different.

\paragraph{Applications.} IVP is a statistical result, which we believe is of intrinsic interest. Our paper shows why it is also instrumentally useful in games of asymmetric information.
After all, even in common-prior environments, private information can endow Anne, an informed agent, with a different belief about a fundamental than Bob, an uninformed agent. In equilibrium, Anne anticipates how her actions will affect Bob's belief. Anne's strategic incentives may depend on how she expects new information to affect Bob, for which IVP is a useful tool.  The new information can be exogenous or endogenous, e.g., owing to the behavior of still other agents.

We develop these points in two contexts.

\paragraph{\emph{Voluntary testing}.}  \autoref{sec:testing} studies voluntary testing or certification. An agent has some private information about his true ability or product quality.  He can choose to undertake a costly test that provides an independent public signal of quality. Gross of the testing cost, the agent's payoff is the market's posterior expectation of his quality.

\autoref{prop:testing} establishes that under familiar informational assumptions, the more (Blackwell) informative the test, the less the agent will choose to get tested. The logic we elucidate, using IVP, is that because of pooling, the marginal type who takes the test expects more informative tests to reduce the benefit of getting tested.

\autoref{prop:testing} offers economic insights. When market information is of concern, a tradeoff must be resolved between more informative tests and participation \citep[cf.][]{HR18}. But in settings where information only affects surplus division, ex-ante efficiency improves with better tests because they reduce the deadweight loss from testing. We further discuss implications for a monopolist certifier's test choice, generalizing \citepos{Lizzeri99} observation about the emergence of uninformative tests.

\paragraph{\emph{Costly signaling.}}
IVP is also useful in games of asymmetric information even when the information asymmetry is eliminated in equilibrium by an agent's own behavior. \autoref{sec:lying} deploys IVP in a sender-receiver game with lying costs \citep{Kartik09}, but explains how the logic also applies to canonical applications like education signaling \citep{Spence73}. Our concern is how exogenous information the receiver obtains affects the sender's signaling.

\autoref{prop:efficiency} establishes that, under reasonable conditions, better exogenous information reduces the sender's benefit from falsification to appear more favorable.  Intuitively, the sender expects any favorable receiver belief he induces to get neutralized more by better exogenous information.  Costly falsification becomes less attractive. Consequently, better exogenous information reduces wasteful signaling---every sender type is better off, even under full separation. We further discuss how this implies a strategic complementary when there are multiple, possibly opposed, senders with bounded signal spaces.

\paragraph{Other applications.} We believe IVP will be useful in other contexts too.  Indeed, the logic of IVP unifies aspects of mechanisms in some existing papers that study models with heterogeneous priors under specific information structures. 
See, for example, \citet{Yildiz04}, \citet{CK09}, \citet[][Proposition 5]{VdS10},  \citet[][Proposition 8]{Hirsch13}, and \citet[][Proposition 5]{SY12}.
We ourselves have used IVP to study information acquisition prior to disclosure \citep{KLS17-Acquisition}.

\section{Information Validates the Prior}
\label{sec:IVP}

Fix any finite (multi-)set of states $\Omega\equiv \{\omega_1,\ldots,\omega_L\} \subset \mathbb{R}$, with $\omega_1\leq \ldots\leq \omega_L$. We denote a generic element of $\Omega$ by either $\omega$ or $\omega_l$.  Fix a measurable space of signals, $(S,\mathcal S)$, endowed with a $\sigma$-finite reference measure. An \emph{experiment} is $\mathcal E\equiv \{P_\omega\}_{\omega\in \Omega}$, where $P_\omega$ is a probability measure over signals in state $\omega$. We only consider experiments for which each $P_\omega$ is absolutely continuous with respect to the reference measure, so that a Radon-Nikodym derivative exists, denoted $p(s|\omega)$; this is the probability density or mass function.\footnote{This permits the familiar definition of monotone likelihood ratio below. \citepos{LW20} ``strong stochastic dominance'' should be applicable more generally.}

Experiment $\mathcal E$ is (Blackwell) \emph{more informative} than experiment $\tilde {\mathcal E} \equiv \{\tilde P_\omega\}_{\omega\in \Omega}$,  whose signal may be denoted $\tilde s\in S$ for clarity, 
if there is a Markov kernel  $Q(\cdot| s)$ such that for each $\omega\in \Omega$ and every $\Sigma\in \mathcal{S}$, 
$$\tilde P_\omega(\Sigma)=\int_{S}Q(\Sigma| s)\, \mathrm d P_\omega(s).$$
This definition captures the statistical notion that $\cal{\tilde E}$ does not provide any information beyond $\mathcal E$: state by state, the distribution of signals in $\cal{\tilde E}$ can be generated by taking signals from $\cal E$ and transforming them through the state-independent kernel $Q(\cdot)$. As established by \citet{Blackwell53}, this notion is equivalent to that of $\mathcal E$ being more valuable than $\mathcal{\tilde E}$ for decision making in canonical senses.

An experiment is an \emph{MLRP-experiment} if there is a total order on $S$, denoted $\succeq$ (with asymmetric relation $\succ$), such that the monotone likelihood-ratio property (MLRP) holds:
$$s' \succ s\text{ and } \omega'>\omega \implies p(s'|\omega') p(s|\omega)\geq p(s'|\omega) p(s|\omega').$$

Let $\beta \in \Delta \Omega$ denote a belief, with $\beta(\omega)$ the probability ascribed to state $\omega$.   We say that a belief $\beta'$ \emph{likelihood-ratio (LR) dominates} belief ${\beta}$, written ${\beta}\leq_{LR} {\beta'}$, if for all $\omega' > \omega$, 
\begin{equation*}
 \beta'(\omega')\beta(\omega) \geq \beta(\omega')\beta'(\omega).
\end{equation*}
A pair of beliefs are \emph{likelihood-ratio ordered} if one likelihood-ratio dominates the other.

Anne ($A$) and Bob ($B$) are Bayesians with full-support priors on $\Omega$,
denoted by 
$\beta_A$ and ${\beta}_B$ respectively, with expectations or means $m_A$ and $m_B$.\footnote{\label{fn:rescale}That is, $m_i \equiv \sum_{\omega} \omega \beta_i(\omega)$. We could just as well take $m_i$ to be $\sum_\omega h(\omega)\beta_i(\omega)$ for any increasing $h:\Omega\to \Reals$. (By ``increasing'', we mean ``weakly increasing''; similarly for related terminology throughout this paper, unless made explicit otherwise.) This amounts to relabeling each state $\omega$ as $h(\omega)$.} Given an experiment $\mathcal E$, we denote individual $i$'s posterior mean by ${m}_i^s$, computed by Bayes rule. Let $\E_i^{\mathcal E}[m^s_j]$ denote $i$'s ex-ante expectation of $j$'s posterior mean.

If the individuals have the same prior, then $\E_A^{{\mathcal E}} [m^s_B]=\E_A^{{\mathcal E}} [m^s_A]=m_A$ by iterated expectation. Under different priors, the same conclusion holds when $\mathcal E$ is fully informative (i.e., every signal reveals the state) because prior differences become irrelevant. On the other hand, if $\mathcal E$ is  uninformative (i.e., no signal provides any information), then $\E_A^{{\mathcal E}} [m^s_B]=m_B$. For the intermediate cases between fully informative and uninformative experiments, there is certain monotonicity under some conditions:
\begin{theorem}
\label{thm:IVP}
Let ${\beta}_A\leq_{LR} \beta_B$, and $\mathcal E$ and $\tilde{\mathcal E}$ be MLRP-experiments. If ${\mathcal E}$ is more informative than $\tilde{\mathcal E}$, then:
\vspace{-5pt}
\begin{align*}
&m_A\leq \E_A^{\mathcal E} \left[m_B^s\right]\leq \E_A^{\tilde{\mathcal E}} \left[m^{\tilde s}_B\right]\leq m_B; \ \ {and} \\[5pt]
&m_A\leq \E_B^{\tilde{\mathcal E}} \left[m^{\tilde s}_A\right]\leq \E_B^{\mathcal E} \left[m_A^s\right]\leq m_B.
\end{align*}
\end{theorem}

\autoref{thm:IVP} says that under its ordering requirements,\footnote{The priors' likelihood-ratio ordering may be viewed as without loss---the states can be relabeled---so long as the experiments' MLRP and beliefs' means are understood with respect to the states' relabeling.} each individual $i$ predicts that a more informative experiment will, on average, bring the other's posterior mean closer to $i$'s prior. In this sense more information is expected to further validate one's prior; or more succinctly, \emph{information validates the prior} (\emph{IVP}). It bears emphasis that by relabeling states (cf. \autoref{fn:rescale}), information validates the prior not only in the sense of the mean state, but the expectation of any increasing function of the state, such as the probability of any $\{\omega_l,\ldots,\omega_L\}$. 

The proof of \autoref{thm:IVP} is in \autoref{app:proofs}. For an illustration, suppose Bayesian updating takes the canonical linear form: the posterior mean is a convex combination of the prior mean and the signal, as is the case for any exponential family of signals with conjugate prior (e.g., normal-normal). It then holds that for any signal $s\in \Reals$ under experiment $\cal E$, $m^s_j=(1-\alpha^{\cal E}) m_j + \alpha^{\cal E} s$ for some $\alpha^{\cal E} \in [0,1]$. Hence, 
\begin{equation*}
\E^{\cal E}_i \left[m^s_j\right]=(1-\alpha^{\cal E}) m_j + \alpha^{\cal E} m_i.
\end{equation*}
The weight $\alpha^{\cal E}$ is larger when experiment $\mathcal E$ is more informative, which implies \autoref{thm:IVP}'s conclusions.

To shed light on the ordering requirements in \autoref{thm:IVP}, it is useful to decompose its conclusions. Each individual $i$ expects the other's posterior mean to:
\begin{enumerate}[label=\roman*.,topsep=0pt,itemsep=3pt]
	\item\label{direction} (\emph{Direction}) move towards	$i$'s prior mean under any experiment: e.g., $\E^{\mathcal E}_A[m^s_B]\leq m_B$;
	\item\label{bound} (\emph{Undershooting}) not move past $i$'s prior mean under any experiment: e.g., $m_A\leq \E^{\mathcal E}_A[m^s_B]$;
	\item\label{intensity} (\emph{Intensity}) move more when the experiment is more informative: e.g., $\E^{\mathcal E}_A[m^{s}_B]\leq \E^{\mathcal {\tilde E}}_A[m^{\tilde s}_B]$. 
\end{enumerate}

Note that the third point implies the previous two given the observations before the theorem about extreme experiments. The directional result (point \ref{direction}) does not require priors to be likelihood-ratio ordered; given the experiment's MLRP, the priors' ordering by first-order stochastic dominance (FOSD) is sufficient, and is in fact essentially necessary.\footnote{Suppose $m_A\leq m_B$ with $\beta_A \nleq_{FOSD} \beta_B$. Then for some index $k$, $\sum_{l\geq k}\beta_B(\omega_l)<\sum_{l\geq k}\beta_A(\omega_l)$. For an MLRP-experiment $\mathcal E$ that only reveals whether $\omega<\omega_k$ or $\omega\geq \omega_k$, it can be verified that $\E^{\mathcal E}_A[m^s_B]>m_B$ so long as $\omega_1<\omega_L$.
} 
\suppappendixref{app:tightness} shows that even with likelihood-ratio ordered priors, the directional result can fail for a non-MLRP-experiment. The undershooting result (point \ref{bound}) does not require any assumption on the experiment; the priors' likelihood-ratio ordering is sufficient, and is in fact essentially necessary.\footnote{\label{fn:FOSD} If $\beta_A \nleq_{LR} \beta_B$, then so long as $\omega_1<\ldots<\omega_L$ there is an MLRP-experiment $\mathcal E$ such that 
$\E_A^{{\mathcal E}} [m^s_B]<m_A$.
Specifically, let $l$ be any index such that $\frac{\beta_B(\omega_{l+1})}{\beta_B(\omega_l)}<\frac{\beta_A(\omega_{l+1})}{\beta_A(\omega_l)}$ and consider $\mathcal E$ that fully reveals every state except $\{\omega_l,\omega_{l+1}\}$, which are pooled together. \citet[Proposition 4]{OR19} note a related point.} The intensity result (point \ref{intensity}) does not require any assumption on the less informative experiment; the priors' likelihood-ratio ordering and the more informative experiment's MLRP are sufficient. This is verified by studying the theorem's proof, which is fairly straightforward because the experiments are Blackwell-ranked. We leave to future research the question of whether less restrictive comparisons of experiments would suffice.

\citet{FK14} prove that conditional on any event being true, a Bayesian's expected posterior on that event (ignorant of the truth) is larger than her prior; see also \citet{Good65}. It can be shown that for a binary state, that result and \autoref{thm:IVP} are equivalent. Note that any priors and experiments satisfy the theorem's ordering assumptions with a binary state. More generally though, neither result implies the other.

IVP has a noteworthy implication about expected disagreement. Consider measuring disagreement between beliefs by the distance in their means.\footnote{The distance between individuals'
expectations is an interesting measure of disagreement, but obviously coarse and not without limitations.
\citet{Zanardo17} studies disagreement axiomatically; he provides a result in the spirit of our \autoref{cor:disagree} for a family of disagreement measures that includes the Kullback-Leibler divergence.} Typically, experiments can generate signals for which posterior disagreement is larger than prior disagreement.\footnote{Yet no signal can polarize beliefs in the sense of FOSD \citep[][Theorem 1]{BHK13}.}
But can individuals expect, ex ante, to disagree more after observing more information?

\begin{corollary}
\label{cor:disagree}	
Let ${\beta}_A\leq_{LR} {\beta}_B$, and $\mathcal E$ and $\tilde{\mathcal E}$ be MLRP-experiments. If ${\mathcal E}$ is more informative than $\tilde{\mathcal E}$, then for either $i$ and $j$,
\begin{equation*}
\E_i^{\mathcal E} \left[\abs{m^s_i-m^s_j}\right] \leq \E_i^{\tilde{\mathcal E}} \left[\abs{m^{\tilde s}_i-m^{\tilde s}_j}\right].
\end{equation*}
\end{corollary}

\autoref{cor:disagree} says that---subject to the ordering hypotheses---more information reduces expected disagreement when disagreement is measured by the (absolute) difference in means. We omit a proof as the corollary is equivalent to \autoref{thm:IVP} because any signal from an MLRP-experiment preserves the prior likelihood-ratio ordering of beliefs and thus also the ordering by their means; moreover, for any experiment $\cal E$ and individual $i$, $\E^{\mathcal E}_i[m^s_i]=m_i$.

\autoref{cor:disagree} and IVP relate to the literature on merging of opinions initiated by \citet{BD62}. That literature establishes asymptotic merging under general conditions; we are concerned instead with updating after just one observation.

In the following sections, we apply IVP  to study games with a common prior but asymmetric information. 

\section{Voluntary Testing}
\label{sec:testing}

\paragraph{Model.} An agent has unknown ability or product quality $q\in Q\subset \Reals$, where $Q$ is a finite set. There is a full-support prior $\pi(q)$, whose expectation is denoted $\pi^e$. The agent privately receives information, referred to as his type, $t\in T\equiv [\underline t,\overline t]\subset \Reals$ from the density $f(t|q)$ satisfying the strict MLRP. 
Without loss, we assume $\E[q|t]=t$. The agent then chooses whether to take a test or not. If he takes the test, a result or signal $s \in S$ is drawn from density $g(s|q)$ with MLRP; $s$ is conditionally independent of $t$. After observing whether the agent took the test, and if so, the test result, a decision maker or market forms belief $\delta \in \Delta Q$ and hires the agent at a wage (or buys the product at a price) $w=\delta^e$, the market's posterior expectation of quality. Taking the test costs the agent $c\geq 0$, and so the agent's von Neumann-Morgenstern payoff is $w-c$ if he took the test and $w$ if he did not.

\paragraph{Equilibrium.}

Since the test is conditionally independent of the agent's information, we can think of the market as updating in two steps: it first forms an ``interim belief'' about quality by updating its prior with what the agent's action reveals about his private information $t$; thereafter, if the agent took the test, it updates the interim belief using the test result $s$. Accordingly, denote the market's interim belief when the agent takes the test  (before observing the test result) as $\delta_+$ and when the agent does not as $\delta_-$. The agent's private belief is denoted $\beta(t)$. Plainly, the agent of type $t$ will take the test if and only if (modulo indifference)
$$\E_{s|t}[\delta^e(s;\delta_+)]-c > \delta^e_-,$$
where $\delta^e(s;\delta_+)$ is the market posterior expectation upon observing test result $s$ given interim belief $\delta_+$, and $\delta^e_-$ is the no-test expected quality. The left-hand side of the above inequality is strictly increasing in $t$: for any $t_H>t_L$,
\begin{align*}
\E_{s|t_H}[\delta^e(s;\delta_+)]-\E_{s|t_L}[\delta^e(s;\delta_+)]=&\sum_{q}\left[\beta(q|t_H)-\beta(q|t_L)\right]\E_{s|q}[\delta^e(s;\delta_+)]
\end{align*}
is strictly positive because $\E_{s|q}[\delta^e(s;\delta_+)]$ is increasing in $q$ by the test's MLRP, and $\beta(\cdot|t_H)$ strictly first-order stochastically dominates $\beta(\cdot|t_L)$ by strict MLRP of the agent's private information.\footnote{More precisely, our conclusion holds if $\E_{s|q}[\delta^e(s;\delta_+)]$ is strictly increasing in $q$. Accounting for the possibility that this function may only be weakly increasing would not materially affect what follows. The possibility is ruled out when $g(s|q)$ has the strict MLRP and $\delta_+$ is non-degenerate.}

Therefore, all (weak Perfect Bayesian) equilibria are described by a cutoff $t^*\in T$ such that the agent takes the test if $t>t^*$ and does not if $t<t^*$. Let $\delta_-(t^*)$ and $\delta_+(t^*)$ denote the interim beliefs computed from Bayes rule given any interior cutoff $t^*$. Consistent with the expected gain from testing strictly increasing in type, we restrict attention to equilibria in which if no type gets tested ($t^*=\bar t$), then the off-path interim belief upon testing is $\delta_+(\bar t)\equiv \beta(\bar t)$; similarly, if all types get tested ($t^*=\underline t$), then the no-test off-path belief is $\delta_-(\underline t)\equiv \beta(\underline t)$. Since intervals of signals preserve MLRP structure \citep[][Theorem 4]{Milgrom81}, it holds that for any $t^* \in T$, $\delta_-(t^*)$ and $\delta_+(t^*)$ are likelihood-ratio ordered with all $\beta(t)$.

\begin{lemma}
\label{lem:testing_eqm}
A cutoff $t^*$ is an equilibrium cutoff if and only if:
\begin{enumerate}
\item $t^*$ is interior and 
\begin{equation}
\label{e:testing_eqm}
\E_{s|t^*}[\delta^e(s;\delta_+(t^*))]=\underbrace{\E[t|t<t^*]}_{=\delta^e_-(t^*)}+c;
\end{equation}

\item $t^*=\underline t$ and $\E_{s|\underline t}[\delta^e(s;\pi)]\geq \underline t+c$; or
\item \vspace{5pt}  $t^*=\bar t$ and $\bar t\leq \pi^e+c$.
\end{enumerate}
\end{lemma}

We omit the routine proof. The characterization implies that an equilibrium exists. There may be multiple equilibria. A key observation is that at any interior equilibrium cutoff $t^*$, 
$$\delta^e_-(t^*)=\E[t|t<t^*]<t^*\leq \E_{s|t^*}[\delta^e(s;\delta_+(t^*))]\leq \delta^e_+(t^*)=\E[t|t>t^*],$$
with at least one of the weak inequalities being strict. The weak inequalities follow from IVP, as \autoref{thm:IVP}'s ordering hypotheses are satisfied. Intuitively, as the cutoff type is pooled with all higher types, it expects the test to cause a downward revision from the market's interim quality expectation. The cutoff type is nevertheless willing to take the test because foregoing it would lead to an even worse expectation, $\delta^e_-(t^*)$. It follows that $t^*-\delta^e_-(t^*)\leq c$; equality holds only with a fully informative test. With such a test, the setting is effectively one of costly voluntary disclosure \citep{Jovanovic82,Verrecchia83}. If the test is uninformative, then $\delta^e_+(t^*)-\delta^e_-(t^*)=c$; it is as if the agent can simply take a pure money-burning action.
 

\paragraph{Comparative statics.} What happens when the test becomes (Blackwell) more informative? IVP assures that the cutoff type expects the ``disagreement'' between its private belief and the market's interim belief $\delta_+$ to shrink by more. The cutoff type's expected benefit of taking the test is thus lower. On the other hand, the payoff from not taking the test, $\delta^e_-$, is unaffected. Consequently, fewer types should take the test. Formally:

\begin{proposition}
\label{prop:testing}
If the test becomes more informative, then the smallest and largest equilibrium cutoffs increase.
\end{proposition}

\begin{proof}
Consider two 
test
distributions, $\underline g$ and $\overline g$, with the latter more informative. Define $L(t;g) \equiv \E^g_{s|t}[\delta^e(s;\delta_+(t))]$ and $R(t)\equiv \delta^e_-(t)+c$; these correspond to the two sides of \autoref{e:testing_eqm}. IVP (\autoref{thm:IVP}) implies that $L(t;\underline g)\geq L(t;{\bar g})$ for all $t$. We argue below that the largest equilibrium cutoff increases; an analogous argument applies to the smallest.

The conclusion is trivial if $\bar t \leq \pi^e+c$, as \autoref{lem:testing_eqm} implies the largest equilibrium cutoff is $\bar t$ regardless of the test. So suppose $\bar t >\pi^e+c$, or equivalently, $L(\bar t;g)>R(\bar t)$ for any $g$. If $L(t;{\underline g})>R(t)$ for all $t>\underline t$, the only equilibrium cutoff under $\underline g$ is $\underline t$, and the conclusion follows. So suppose $L(t;{\underline g})=R(t)$ for some interior $t$. Continuity and $L(\cdot;\underline g)\geq L(\cdot;\overline g)$ imply that the largest intersection of $L(\cdot;\overline g)$ and $R(\cdot)$ is at least as large as that of $L(\cdot;\underline g)$ and $R(\cdot)$.
\end{proof}

Equilibria with extremal cutoffs are stable in the sense of best-response dynamics. 
\autoref{prop:testing}'s comparative static extends to other stable equilibria, but it can reverse for unstable equilibria, as is a common theme in games with multiple equilibria.

Ex ante, the agent prefers a larger equilibrium cutoff because his ex-ante expected wage is $\pi^e$ in any equilibrium, and a larger cutoff reduces the likelihood of having to take the test. Focussing on the largest equilibrium is thus justified on ex-ante utilitarian grounds if the agent's wage is purely redistributive and information has no efficiency benefit. Indeed, \autoref{prop:testing} implies that from this perspective more informative tests and their largest equilibrium cutoffs are socially preferable---only indirectly, because they lower the deadweight loss from testing.

\autoref{prop:testing}'s lesson that more informative tests reduce testing participation echoes and broadens a point made by \citet{HR18} that fully informative tests do not maximize participation. Unlike them, we do not require the agent to be perfectly informed about quality, and we compare any pair of Blackwell-ranked tests. Their focus was on solving for the optimal test to minimize the market's mean squared error.

\paragraph{Certification by an intermediary.} View the agent's cost $c$ as the price charged by a monopolist testing firm or certifying intermediary who must offer a single test. Momentarily take $c$ as given. The monopolist then seeks to maximize the fraction of agents who take the test, i.e., to minimize the equilibrium cutoff. It follows from \autoref{prop:testing} that (focussing on an extremal equilibrium) the monopolist will choose the least informative test available, if such a test exists---even if all available tests are costless to perform, and \emph{a fortiori}, if more informative tests are more costly. In particular, if the uninformative test is available (and least costly), the monopolist will choose that. Consequently, when the monopolist can choose any test and any price $c$, and we select the profit-maximizing equilibrium cutoff, it is optimal to choose an uninformative test and price $c=\pi^e-\underline t$. All types then get tested and the monopolist makes the maximum possible profit subject to the agent's minimum payoff of $\underline t$.\footnote{If the type distribution under the prior has a strictly decreasing density, then $\delta^e_+(t)-\delta^e_-(t)$ is strictly increasing and hence $\underline t$ is the unique cutoff equilibrium when the test is uninformative and $c=\pi^e-\underline t$.} 

This discussion generalizes themes from \citet{Lizzeri99}, who assumed a perfectly-informed agent, endogenous pricing, and availability of all tests. Our analysis clarifies that the economic force favoring less informative tests does not turn on any of these conditions---rather, even with a partially-informed agent, demand in a profit-maximizing equilibrium at any price is higher when the monopolist chooses a less informative test.

\section{Costly Signaling}
\label{sec:lying}

Our second application shows how IVP is useful in games with asymmetric information even when there may be no ``disagreement'' {in equilibrium} because of (full) separation; rather, what is crucial is how information affects disagreement \emph{off the equilibrium path}. 

\subsection{Model}
We consider communication with lying costs or costly falsification, following \citet{KOS07} and \citet{Kartik09}, but adding exogenous information. A sender and a receiver share a common non-degenerate prior about a state $\omega 
\in \{0,1\}$.
The sender privately learns his type 
$t\in [0,1]$, normalized to equal his private belief that $\omega=1$, drawn from a density $f(t|\omega)$. Our normalization implies that $f(\cdot)$ has the MLRP. The sender sends a report $r\geq 0$ at a cost $c(r,t)$, elaborated below. The receiver forms a belief based on both $r$ and an additional signal $s\in \Reals$ that, conditional on the state $\omega$, is drawn independently of $t$ or $r$ from a density $g(s|\omega)>0$.  Without loss, we assume $g(\cdot)$ satisfies the MLRP. The signal $s$ can either be the receiver's private information or publicly observed after the sender has acted.  Denote the receiver's posterior expectation of the state by $\E [\omega|r,s]$.  The sender's von Neumann-Morgenstern payoff is linear in this expectation; specifically, his payoff is
\begin{equation*}
\E [\omega | r,s] - c(r,t).
\end{equation*}

The receiver's belief updating can be analyzed in two steps: based on the sender's report $r$, she forms an interim belief $\pi(r)\in [0,1]$ about the state (all beliefs refer to the probability of $\omega=1$) and then uses this interim belief to further update from the signal $s$ to form a posterior belief about the state
$\beta(s;\pi(r))$. 
By Bayes rule,
\begin{equation}
\beta(s;\pi)
=\frac{\pi g(s|1)}{\pi g(s|1) + (1-\pi)g(s|0)}.
\label{e:signalingBayes}
\end{equation}
The expected payoff of a type-$t$ sender from report $r$ is thus
\begin{equation*}
\E_{s|t}[
\beta(s;\pi(r))
]-c(r,t),
\end{equation*}
where $\E_{s|t}$ denotes $t$'s expectation over the signal $s$. The first term in the above display is strictly increasing in $\pi(r)$ because a higher interim belief raises the receiver's posterior after any signal.

Assume the cost function $c(r,t)$ is smooth with $\partial c(t,t)/\partial r=0$ for all $t$, i.e., the marginal cost of lying when telling the truth is zero.  Furthermore, $\partial^2 c(r,t)/\partial r^2>0>\partial^2 c(r,t)/\partial r\partial t$ for all $t,r$, i.e., the marginal cost of sending a higher report is increasing in the report and decreasing in the type.

\paragraph{A single-crossing condition.} Due to the receiver's exogenous information, the above cost assumptions do not guarantee a suitable single-crossing property. Consider the sender's indifference curves in the space of his report $r$ and the receiver's interim belief $\pi$. For any type $t$, these indifference curves slope upwards for $r > t$: costly lying requires a higher interim-belief compensation. The requisite single-crossing property
 is that---no matter the exogenous experiment---
 these indifference curves are flatter for higher types, meaning that higher types are more willing to inflate their report to induce a higher interim belief.  \autoref{lem:single-cross} in \autoref{app:proofs} establishes that this property is assured by the following assumption that we will maintain:
\begin{equation}
\text{For $t<1$ and $t<r$:}\quad
\frac{\partial^2 c(r,t)/\partial r \partial t}{\partial c(r,t)/\partial r} \le -\frac{1}{1-t}.
\label{a:SC}	
\end{equation}

Condition \eqref{a:SC} can be interpreted as saying that higher types have a sufficiently large marginal cost advantage relative to marginal cost. All our cost assumptions are satisfied by, for example, $c(r,t)=(r-t)^2$. 

\paragraph{Another interpretation.} Although our model is posed as a communication game, it can also be viewed as adding exogenous information to a variation of the \citet{Spence73} signaling model.  In this interpretation, a worker possesses a private trait, his type $t$ (e.g., intelligence), which is indicative of binary job productivity. An employer also observes a signal about that productivity (e.g., through an interview). The worker's wage depends on the employer's posterior on productivity given the worker's schooling level $r$ and the employer's signal $s$.  The marginal cost of schooling decreases in the characteristic $t$. 
While we assume that higher-type workers intrinsically prefer acquiring more education, 
our analysis also applies if all workers prefer less education; 
see \autoref{sec:signaling_disc}.

\subsection{Equilibrium}

We focus on the least-cost (fully) separating equilibrium, or LCSE, as is standard. In such an equilibrium the sender's pure strategy $\rho:[0,1]\to \Reals_+$ is strictly increasing, with the ``Riley condition'' $\rho(0)=0$. For $r\in [0,\rho(1)]$, the receiver's interim belief upon observing $r$ is $\pi(r)=\rho^{-1}(r)$. Without loss, we stipulate that for $r>\rho(1)$, $\pi(r)=1$. Standard arguments imply that the LCSE strategy $\rho(\cdot)$ must satisfy the differential equation
\begin{equation}
\label{e:DE}
\frac{\partial c(\rho(t),t)}{\partial r} \rho'(t)
= \frac{\partial \E_{s|t}\left[
\beta(s;t)
\right]}{\partial \pi}.
\end{equation}
\autoref{e:DE} obtains from the binding local upward incentive compatibility constraints.
The left-hand side is type $t$'s marginal cost of mimicking a slightly higher type; the right-hand side is the marginal benefit, which comes from inducing a higher receiver interim belief.  This benefit is affected by the sender's belief about the exogenous signal $s$. Since $\rho'(t)>0$ and \autoref{e:DE}'s right-hand side is strictly positive, any solution has $\rho(t)>t$ for $t>0$. 

Condition \eqref{a:SC}, which yields the requisite single-crossing property, guarantees that not only are \autoref{e:DE} and the boundary condition $\rho(0)=0$ necessary in an LCSE, but they are also sufficient (i.e., global incentive compatibility is assured). As standard arguments imply that this boundary-value problem has a unique solution, we state without proof:
\begin{lemma}
\label{lem:uniqueLCSE}
There is a unique LCSE.
\end{lemma}

\subsection{Comparative statics}

How does (Blackwell) more informative exogenous information affect the sender's LCSE signaling strategy? The strategy is determined by the local upward incentive constraints. Intuitively, when type $t$ mimics a slightly higher type $t+\epsilon$ it creates disagreement: type $t$ views the receiver's interim belief $t+\epsilon$ as higher than the truth. IVP implies that $t$ expects a more informative exogenous signal to correct that interim belief by more, and so $t$'s gain from inducing that interim belief is lower under a more informative experiment. Formally:

\begin{lemma}
\label{lem:compare}
If $\tilde s$ is drawn from a more informative experiment than $s$, then for any $t<1$,
$$\frac{\partial \E_{\tilde s |t} \left[
\beta(\tilde s;t)
\right]}{\partial \pi} \leq \frac{\partial \E_{ s | t} \left[
\beta(s;t)
\right]}{\partial \pi}.$$
\end{lemma}

\begin{proof}
IVP (\autoref{thm:IVP}) implies that for any $t<1$ and small $\epsilon>0$,
\begin{equation*}
\E_{\tilde s | t}\left[
\beta(\tilde s; t+\epsilon)
\right] - \E_{\tilde s | t} \left[
\beta(\tilde s; t)
\right] \leq \E_{ s | t}\left[
\beta(s; t+\epsilon)
\right] - \E_{ s | t}\left[
\beta(s; t)
\right],
\end{equation*}
because $\E_{\tilde s | t}[\beta(\tilde s;t)]=\E_{ s | t}[\beta(s; t)
]=t$. (\autoref{thm:IVP}'s ordering hypotheses hold because both the sender's and receiver's information have the MLRP.) The result follows from dividing both sides of the above inequality by $\epsilon$ and taking $\epsilon \to 0$.
\end{proof}

\autoref{lem:compare} implies that more informative exogenous information reduces the right-hand side of \autoref{e:DE}.  Since each type expects a smaller marginal benefit from inducing a higher belief in the receiver, the solution $\rho$ to the differential equation \eqref{e:DE} with boundary 
condition
$\rho(0)=0$ is pointwise lower. It is therefore intuitive, and proved in \autoref{app:proofs}, that:

\begin{proposition}
\label{prop:efficiency}
In the LCSE, every sender type bears a lower signaling cost when the receiver's exogenous information is more informative.
\end{proposition}

Consequently, every sender type is better off when the receiver's information improves, because in the LCSE every type $t$ expects the receiver's posterior to be $t$ regardless of the exogenous information distribution. Plainly, the receiver is also better off (given any von Neumann-Morgenstern payoff function) with better exogenous information.

\subsection{Discussion and implications}
\label{sec:signaling_disc}
A very similar analysis, with the same conclusion as \autoref{prop:efficiency}, would apply if we had instead assumed \`{a} la \citet{Spence73} that all types would choose $r=0$ under complete information. Specifically, we could have instead assumed the signaling cost function $c(r,t)$ satisfies $\partial c/\partial r>0$, $\partial^2 c/\partial r^2>0$, and $\partial^2 c/\partial r\partial t<0$, and required the inequality in \eqref{a:SC} to hold for all $t<1$ and $r>0$.

One reason we study costly lying is that the analysis also extends to a bounded report space, which is natural there. Suppose the sender's report must be in $[0,1]$, his type space. Then, given our original assumptions on $c(\cdot)$, there is no fully separating equilibrium; in particular, recall that for a strictly increasing strategy $\rho$, a solution to \autoref{e:DE} entails $\rho(t)>t$ for all $t>0$. In a previous version of this paper \citep{KLS19-IVP}, we analyzed a salient equilibrium that extends the LCSE to this case: there is separation up to some cutoff type $t^*<1$ and pooling on $r=1$ thereafter. The sender's strategy in the separating region is unchanged: $\rho(0)=0$ and \autoref{e:DE} holds. The cutoff $t^*$ is determined by its indifference. More informative exogenous information now not only lowers every type's signaling cost, but also raises the cutoff $t^*$ (again due to IVP, similar to the logic that reduces the benefit of pooling in \autoref{sec:testing}). The larger separating region means the sender's signaling is more informative.

One can adapt the analysis to study multi-sender signaling when each sender gets a conditionally independent signal. Suppose some senders are upward biased and others are downward biased, and each chooses a report in $[0,1]$.\footnote{A downward-biased sender's payoff is $-\E[\omega|r,s]-c(r,t)$. The relevant strategy for this sender is given by $\rho(1)=1$, the analog of \autoref{e:DE} with the right-hand side's sign flipped, and all types below some cutoff pooling on report $r=0$. For any $t\in (0,1)$, $\rho(t)<t$, i.e., a downward-biased sender deflates his report.} This model applies to competing persuaders with falsification costs in various domains: lobbyists, media, legal parties, etc. From each sender's perspective, the other senders' reports are endogenous experiments. The logic underlying \autoref{prop:efficiency} implies that senders' strategies are \emph{strategic complements}: in equilibria where each sender uses a cutoff strategy as described above, each sender reveals more and bears a lower signaling cost when other senders reveal more. The receiver thus learns more with more senders; importantly, beyond the obvious direct benefit of adding a sender, there is an indirect benefit of existing senders revealing more. Each sender's ex-ante welfare is also higher, as his equilibrium signaling cost decreases while the receiver's average posterior is unaffected. Furthermore, if any sender's cost increases in a suitable sense (for example, if sender $i$'s cost is $k_i c(r,t)$ and $k_i>0$ increases), 
then not only does that sender reveal more, but so do all other senders.

Although it may seem unsurprising that better receiver information reduces a sender's incentive to incur falsification costs, we emphasize that \autoref{prop:efficiency} relies on IVP.  \suppappendixref{app:new} shows that if the sender's payoff is not linear in the receiver's posterior, or if the receiver's information violates MLRP (in a multi-state extension of the model), then better exogenous information can \emph{increase} the sender's marginal benefit from mimicking a higher type, leading to higher equilibrium signaling costs.

\citet[Section III]{Frank85} also suggests that better exogenous information can reduce dissipative  signaling. \citet{Weiss83} studies when exogenous information allows for separating equilibria even absent any heterogeneity in the direct costs of signaling; his focus is not on comparative statics. \citet{DG14} emphasize the stability of non-separating equilibria when there is ``double crossing'' of appropriate indifference curves, contrary to the single crossing assured by our condition \eqref{a:SC}. \citet{Truyts14} shows that better exogenous information can exacerbate dissipative signaling when signaling is noisy.

\section{Conclusion}
\label{sec:conclusion}

IVP is instructive more broadly than for just the two applications with asymmetric information developed in this paper. In an earlier version \citep{KLS19-IVP}, we used IVP to study voluntary disclosure with either concealment or disclosure costs. 

We close by briefly commenting on a domain with symmetric information. Consider Bayesian persuasion \citep{KG11}. \citet{AC16} develop a general analysis under heterogeneous priors; see also \citet{OR19}. When the sender's preferences are state-independent and concave in the receiver's posterior expectation, there is no scope for beneficial persuasion under common priors. \citet[Section 4.3]{AC16} show that this observation does not hold generically with heterogeneous priors and at least three states. 
Our \autoref{thm:IVP} delivers additional insights. For example, if the priors are likelihood-ratio ordered, then among MLRP-experiments the sender prefers less informative experiments when facing a favorable receiver, i.e., one whose prior dominates (resp., is dominated by) the sender's if the sender's utility is increasing (resp., decreasing) in the receiver's posterior expectation. Consequently, in that scenario, an uninformative experiment is optimal if all and only MLRP-experiments are available.\footnote{Following our discussion after \autoref{thm:IVP}, priors ordered by FOSD are sufficient for optimality of an uninformative experiment \citep[also note this point restricting attention to a subset of MLRP-experiments]{OR19}, but not for the broader preference ranking of MLRP-experiments by their informativeness.}

\newpage

\appendix
\numberwithin{equation}{section}
\numberwithin{proposition}{section}
\numberwithin{lemma}{section}
\numberwithin{example}{section}
\numberwithin{observation}{section}
\numberwithin{claim}{section}

\allowdisplaybreaks


\section{Proofs}
\label{app:proofs}

\begin{proof}[Proof of \autoref{thm:IVP}] 
\label{proof:IVP}
We prove the statement about $A$'s expectations; $B$'s are analogous. 
It is sufficient to establish $\E_A^{\mathcal E} \left[m_B^s\right]\leq \E_A^{\tilde{\mathcal E}} \left[m^{\tilde s}_B\right]$, as the other two inequalities are implied by the observations in the paragraph preceding \autoref{thm:IVP}.

Since $\mathcal{E}$ is more informative than $\mathcal{\tilde E}$, there is a joint distribution (more precisely, measure) over signal pairs $(s,\tilde s)$ in each state such that the conditional distribution of $\tilde s$ given $s$ is independent of the state. Fix such a set of joint distributions. Since $\beta_A \leq_{LR} \beta_B$, it follows that for any $\tilde s$, $\beta_A^{\tilde s}\leq_{LR} \beta_B^{\tilde s}$ and hence $\beta_A^{\tilde s}\leq_{FOSD} \beta_B^{\tilde s}$. (Here, $\beta_i^{\tilde s}$ denotes $i$'
s posterior after observing signal $\tilde s$.)
Experiment $\mathcal E$'s MLRP implies that its distribution of signals is increasing in the state in the sense of FOSD. Consequently, conditional on any $\tilde s$, the probability distribution ascribed by $B$ to signals $s$ dominates in FOSD that ascribed by $A$.\footnote{Writing $F_i$ and $F$ for cumulative distributions in the natural way, it holds for any $s$ and $\tilde s$ that $F_i(s|\tilde s)=\sum_{\omega}F(s|\omega) \beta_i(\omega|\tilde s)$, where we can write $F_i(s|\omega)$ rather than $F_i(s|\omega,\tilde s)$ because $\mathcal E$ is more informative than $\mathcal{\tilde E}$. $F_i(s|\tilde s)$ is thus the expectation of the decreasing function $F(s|\cdot)$ with a distribution that is dominated in FOSD for $i=A$ as compared to $i=B$.} As the mean $m^s_B$ is increasing in $s$ (by $\mathcal E$'s MLRP), it follows that  
for any $\tilde s$,
\begin{equation}
\E_A^{\mathcal E} \left[m_B^s \mid \tilde{s}\right]
\leq \E_B^{\mathcal E} \left[m_B^s \mid \tilde{s}\right].
\label{e:condmean}	
\end{equation}
Experiment $\mathcal E$ being more informative than $\mathcal{\tilde E}$ implies that for any $s$ and $\tilde s$, $m_B^s=m_B^{s,\tilde s}$, where $m^{s,\tilde s}_B$ denotes $B$'s posterior mean having observed $(s,\tilde s)$. As $\E_B^{\mathcal E} [m_B^{s,\tilde s} \mid \tilde{s}]=m^{\tilde s}_B$ by iterated expectation, it follows from \eqref{e:condmean} that for any $\tilde s$,
\begin{equation}
\E_A^{\mathcal E} \left[m_B^s \mid \tilde{s}\right]
\leq m_B^{\tilde s}.
\label{e:first}
\end{equation}
Taking the expectation in \eqref{e:first} over $\tilde{s}$ using the prior ${\beta}_A$ yields
\begin{align*}
\E_A^{\mathcal E} \left[m_B^s \right]=\E_A^{\tilde{\mathcal E}}\left[ \E_A^{\mathcal E} \left[m_B^s \mid \tilde{s}\right] \right]
\leq
\E_A^{\tilde{\mathcal E}}\left[m_B^{\tilde{s}} \right]
,
\end{align*}
where the equality is by iterated expectation.
\end{proof}

\bigskip

\begin{lemma}
\label{lem:single-cross}
For $t<1$ and $t<r$, $\dfrac{\partial c(r,t)/\partial r}{\partial \E_{s|t} [
\beta(s;\pi)
]/\partial \pi}$
strictly decreases in $t$.
\end{lemma}

\begin{proof}
Differentiating \eqref{e:signalingBayes} and manipulating, the sender's marginal rate of substitution between $r$ and $\pi$ is given by
\begin{equation*}
\frac{\partial c(r,t)/\partial r}{\partial \E_{s|t} [
\beta(s;\pi)
] /\partial \pi}
=\pi(1-\pi)\frac{\partial c(r,t)/\partial r}{\E_{s|t}[
\beta(s;\pi)
(1-
\beta(s;\pi)
)]}.
\end{equation*}
Differentiating with respect to $t$, noting that $\E_{s|t}[\cdot]=t[\cdot]+(1-t)[\cdot]$, and rearranging, we see that for $t < m$, the marginal rate of substitution strictly decreases in $t$ if
\begin{equation*}
\frac{\partial^2 c(r,t)/\partial r\partial t}{\partial c(r,t)/\partial r} < \frac{\E_{s|1}[\beta(s;\pi)(1-\beta(s;\pi))]-\E_{s|0}[\beta(s;\pi)(1-\beta(s;\pi))]}{t\E_{s|1}[\beta(s;\pi)(1-\beta(s;\pi))]+(1-t)\E_{s|0}[\beta(s;\pi)(1-\beta(s;\pi))]}.
\end{equation*}
Since both
$\E_{s|1}[\beta(s;\pi)(1-\beta(s;\pi))]$ and $\E_{s|0}[\beta(s;\pi)(1-\beta(s;\pi))]$ are strictly positive, for $t<1$ the right-hand side of the above inequality is strictly greater than
\begin{align*}
\frac{0-\E_{s|0}[\beta(s;\pi)(1-\beta(s;\pi))]}{0 +(1-t)\E_{s|0}[\beta(s;\pi)(1-\beta(s;\pi))]} 
= -\frac{1}{1-t}.
\end{align*}   
So, if \eqref{a:SC} holds, the marginal rate of substitution strictly decreases in  $t$ for $t<1$ and $t < r$.
\end{proof}

\bigskip

\begin{proof}[Proof of \autoref{prop:efficiency}]
If we show that $\rho(t)$ decreases pointwise when the right-hand side
of \autoref{e:DE} decreases for all $t$, then the result follows from \autoref{lem:compare}. Accordingly, let $\tilde\rho(t)$ and $\rho(t)$ be two solutions to \autoref{e:DE}, with $\tilde\rho(0)=\rho(0)=0$, where $\tilde\rho$ solves \autoref{e:DE} with a pointwise lower right-hand side.
For any $t>0$, if $\tilde\rho(t)=\rho(t)$ then $\rho'(t)\geq \tilde\rho'(t)>0$.  This
implies that at any touching point, $\rho$ must touch $\tilde\rho$
from below.  Consequently, by continuity,
\begin{equation*}
\rho(t')\geq \tilde\rho(t') \text{ for } t'>0 \implies \rho(t)\geq \tilde\rho(t) \text{ for all } t\geq t'.
\end{equation*}
Now suppose, to contradiction, that $\tilde\rho(\hat t) > \rho(\hat t)$ for some for $\hat t > 0$.  It must hold that $\tilde\rho(t) > \rho(t)$ for all $t\in (0,\hat t)$.  Since $\partial^2 c/\partial r^2>0$ and $\tilde\rho$ corresponds to a lower right-hand side of \autoref{e:DE}, it follows from \autoref{e:DE} that $\tilde\rho'(t)< \rho'(t)$ for all $t\in (0,\hat t)$.  But then
\begin{equation*}
\tilde\rho(\hat t)- \rho(\hat t)=\int_{0}^{\hat t}\left[\tilde\rho'(t)-\rho'(t)\right]\, \mathrm{d}t < 0,
\end{equation*}
a contradiction.
\end{proof}

\newpage

\singlespacing 
\bibliographystyle{econ-aea-NK}
\bibliography{IVP}

\newpage


\section{Supplementary Appendix (Not for Publication)}
\label{app:supp}

\subsection{Discussion of \autoref{thm:IVP}}
\label{app:tightness}

The following example shows that even with likelihood-ratio ordered priors, the ``direction'' portion of \autoref{thm:IVP} can fail with a non-MLRP experiment.

\begin{example}
\label{eg:non-MLRP-exp}
Let $\Omega=\{\omega_1,\omega_2,\omega_3\}$, where $\omega_1<\omega_2<\omega_3$.  Consider the following non-MLRP experiment $\mathcal E$ with signal space $S=\{s_1,s_2\}$:
\begin{align*}
\begin{bmatrix}
    0 & 1 & 0 \\
    1 & 0 & 1 \\
  \end{bmatrix},
\end{align*}
where the entry in row $r$ and column $c$ is $\Pr(s_r|\omega_c)$.
Consider priors ${\beta_A}=(1,0,0)<_{LR} \beta_B=(0,x,1-x)$ for any $x\in (0,1)$. Plainly, $m^{s_2}_B=\omega_3$, and  hence $\omega_1=m_A<m_B<\omega_3=\E^{\mathcal E}_A[m^s_B]$. \finishpoint
\end{example}

The priors in \autoref{eg:non-MLRP-exp} violate full support, but the point goes through if $\beta_A$ and $\beta_B$ are perturbed to satisfy full support. \citet[pp.~674--675]{AC16} use a similar example to illustrate how a ``skeptic'' can design information to persuade a ``believer.''

\bigskip

\subsection{The role of linearity and MLRP in the signaling application}
\label{app:new}

Consider the costly signaling application from \autoref{sec:lying}. Recall that in the LCSE, the sender's strategy $\rho(\cdot)$ is determined by the initial condition $\rho(0)=0$ and the differential equation \eqref{e:DE}:
\begin{equation*}
\frac{\partial c(\rho(t),t)}{\partial r} \rho'(t)
= \frac{\partial \E_{s|t}\left[
\beta(s;t)
\right]}{\partial \pi}. 
\end{equation*}
\autoref{lem:compare} established that
\begin{equation}
\frac{\partial \E_{\tilde s |t} \left[
\beta(\tilde s;t)
\right]}{\partial \pi } \leq \frac{\partial \E_{ s | t} \left[
\beta(s;t) \right]}{\partial \pi}
\label{e:rhsDE}
\end{equation}
when signal $\tilde{s}$ is more informative (i.e., drawn from a more informative experiment) than signal $s$.  Inequality \eqref{e:rhsDE} implies that the solution to the aforementioned initial-value problem is pointwise lower under the more informative experiment, and hence the equilibrium signaling level $\rho(t)$ is lower for every type when the receiver has access to $\tilde{s}$ rather than $s$.

We show below how the conclusion can be altered by dropping either linearity of the sender's payoff in the receiver's posterior (\autoref{eg:nonlinearsignaling}) or the MLRP of the receiver's experiments (\autoref{eg:nonMLRPsignaling}).

\begin{example}
\label{eg:nonlinearsignaling}
Letting $V(\beta)\equiv \beta/(1-\beta)$, suppose the sender's payoff is 
$$V(\beta)-c(r,t),$$
which is convex in the receiver's posterior $\beta$. Condition \eqref{a:SC} in \autoref{sec:lying} continues to imply the relevant single-crossing condition for this modified objective. Using Bayes rule, we compute
\begin{equation*}
\E_{s|t}[V(\beta(s;\pi))] = \frac{\pi}{1-\pi}\E_{s|t}\left[\frac{g(s|1)}{g(s|0)}\right].
\end{equation*}
Differentiating and evaluating at $\pi=t$,
\begin{equation*}
\frac{\partial \E_{s|t}[V(\beta(s;t))]}{\partial \pi} = 
 \frac{1}{t(1-t)}\E_{s|t}\left[\frac{\beta(s;t)}{1-\beta(s;t)}\right].
\end{equation*}
The term inside the expectation operator on the right-hand side above is a convex function of $\beta(\cdot)$.  It follows that
\begin{equation*}
 \frac{\partial \E_{\tilde{s}|t}[V(\beta(\tilde{s};t))]}{\partial \pi} \ge \frac{\partial \E_{s|t}[V(\beta(s;t))]}{\partial \pi},
\end{equation*}
by contrast to \eqref{e:rhsDE}.  That is, the convexity in $V(\cdot)$ is strong enough to ensure that 
the marginal benefit from inducing a higher interim belief $\pi$ (locally, at $\pi=t$) is higher when the exogenous signal is more informative. It follows that in the LCSE, all types bear a \emph{higher} signaling cost when the exogenous signal is more informative.\footnote{On the other hand, if $V(\beta)\equiv \log[\beta/(1-\beta)]$, then the local marginal benefit of inducing a higher interim belief is independent of the exogenous experiment. The reason is that $V(\beta(s,\pi))=\log\left(\frac{\pi}{1-\pi}\right)+\log\left(\frac{g(s|1)}{g(s|0)}\right)$ and hence $\partial \E_{s|t}[V(\beta(s;t))]/\partial \pi$ does not depend on $g(\cdot)$. Note that $V(\cdot)$ here is neither convex nor concave.} \finishpoint
\end{example}

\begin{example}
\label{eg:nonMLRPsignaling}
To see that MLRP-experiments are important, we have to modify the signaling model of \autoref{sec:lying} by introducing more states, because any experiment in a two-state model satisfies MLRP.

Assume a full-support common prior about the state $\omega \in \{0,1, 2\}$. The sender receives some private information, indexed by $t \in [0, 1]$, which updates his belief about the state to $(z, 1-z(1+t), zt)$, where each element of this vector is the probability assigned to the corresponding state. The parameter $z\in (0,1/2)$ is a commonly-known constant.  We refer to $t$ as the sender's type. Letting $M(\beta)\equiv \sum_{\omega} \omega \beta(\omega)$ be the receiver's expectation of the state when she holds belief $\beta$, the sender's payoff is 
$$M(\beta) - c(r,t).$$

Let $s$ represent the outcome of an uninformative experiment, and let ${\beta}^s_{\hat t}$ represent the posterior of the receiver after observing $s$ when she puts probability one on the sender's type $\hat t$.  It clearly holds that 
\begin{equation*}
\E_{s|t}\left[M({\beta}^s_{\hat t})\right] = M(\beta^s_{\hat t}) =  1-z+z\hat t.
\end{equation*}
The derivative with respect to $\hat t$, evaluated at $\hat t=t$, is
\begin{equation}
\frac{\partial \E_{s|t} \left[M({\beta}^s_ t)\right]}{\partial \hat t}	= z.
\label{e:opposite_1}
\end{equation}

Now consider an informative experiment with a binary signal space, $\tilde{s} \in \{l,h\}$. Let the probability distributions $g(\tilde{s}|\omega)$ be given by:
\begin{align*}
\begin{bmatrix}
g(l|0) & g(l|1) & g(l|2) \\
g(h|0) & g(h|1) & g(h|2) \\
  \end{bmatrix}
= \begin{bmatrix}
    0 & 1 & 0 \\
    1 & 0 & 1 \\
  \end{bmatrix}.
\end{align*}
This experiment is the same as that in \autoref{eg:non-MLRP-exp} of  \suppappendixref{app:tightness}; it does not have the MLRP.  

Suppose the receiver ascribes probability one to the sender's type $\hat t$.  By Bayes rule, if the signal realization is $\tilde s=l$, the receiver's posterior is ${\beta}^l_{\hat t}=(0,1,0)$, with $M(\beta^l_{\hat t})=1$.  For signal realization $\tilde s=h$,
\begin{equation*}
{\beta}^h_{\hat t}= \left(\frac{1}{1+\hat t},0,\frac{\hat t}{1+\hat t}\right), \quad \text{with} \quad M({\beta}^h_{\hat t}) = \frac{2\hat t}{1+\hat t}.
\end{equation*}
The sender of type $t$'s expectation is
\begin{equation*}
\E_{\tilde{s}|t} \left[M({\beta}^{\tilde s}_{\hat t})\right] = (1-z(1+t)) M(\beta^l_{\hat t}) + z(1+t) M(\beta^h_{\hat t}). 
\end{equation*}
The derivative with respect to $\hat t$, evaluated at $\hat t=t$, is
\begin{equation}
\frac{\partial \E_{\tilde{s}|t} \left[M({\beta}^{\tilde{s}}_t)\right]}{\partial \hat t} =\frac{2z}{1+ t}.
\label{e:opposite_2}
\end{equation}
Combining \eqref{e:opposite_1} and \eqref{e:opposite_2},
\begin{equation}
\frac{\partial \E_{\tilde{s}|t} \left[M({\beta}^{\tilde{s}}_t)\right]}{\partial \hat t}	\geq \frac{\partial \E_{s|t} \left[M({\beta}^s_ t)\right]}{\partial \hat t},
\label{e:opposite_3}
\end{equation}
which is the opposite inequality to \eqref{e:rhsDE}, even though $\tilde s$ is drawn a more informative experiment than $s$.

In this example, both $\E_{s|t}[M(\beta(s;\hat t))]$ and $\E_{\tilde{s}|t} [ M(\beta(\tilde{s};\hat t))]$ are supermodular in the sender's type $t$.  The assumption that $\partial c(r,t)/\partial r\partial t < 0$ ensures that indifference curves in the space of $(r,\hat t)$ for different types are single crossing.  As local incentive compatibility then implies global incentive compatibility, \eqref{e:opposite_3} implies that in the LCSE all types incur higher signaling costs when the receiver has access to the more informative experiment. \finishpoint
\end{example}

\end{document}